\documentclass[preprint,review,12pt]{elsarticle}

\usepackage{graphicx}
\usepackage{amssymb}
 
\usepackage{amsmath}
\usepackage{dsfont}
\usepackage{bm}
\usepackage{amsthm}

\def\R{{\mathds R}}
\def\C{{\mathds C}}

\newcommand{\be}{\begin{equation}}
\newcommand{\ee}{\end{equation}}

\usepackage{xcolor}

\newcommand{\bzero}{{\mbox{\boldmath $0$}}}
\newcommand{\bI}{{\mbox{\boldmath $I$}}}
\newcommand{\bz}{{\mbox{\boldmath $z$}}}
\newcommand{\bn}{{\mbox{\boldmath $n$}}}
\newcommand{\bme}{{\mbox{\boldmath $m$}}}
\newcommand{\bv}{{\mbox{\boldmath $v$}}}
\newcommand{\bp}{{\mbox{\boldmath $p$}}}

\newcommand{\bx}{{\mbox{\boldmath $x$}}}
\newcommand{\by}{{\mbox{\boldmath $y$}}}
\newcommand{\bor}{{\mbox{\boldmath $r$}}}
\newcommand{\bt}{{\mbox{\boldmath $t$}}}

\newcommand{\bC}{{\mbox{\boldmath $C$}}}

\newcommand{\bD}{{\mbox{\boldmath $D$}}}

\newcommand{\bS}{{\mbox{\boldmath $S$}}}

\def\cC{\mbox{$\mathcal C$}}
\def\cF{\mbox{$\mathcal F$}}

\newcommand{\test}{\mbox{$
\begin{array}{c}
\stackrel{ \stackrel{\textstyle H_1}{\textstyle >} }{ 
\stackrel{\textstyle \leq \phantom{|^|}\!\!\!\!\!\!}{ \textstyle  H_0} }
\end{array}
$}}

\newtheorem{myteo}{Proposition}

\newcommand*{\point}{\makebox[1ex]{\textbf{$\cdot$}}}%

\begin{document}

\title{A $k$-nearest neighbors approach to the design of radar detectors}

\author[a]{Angelo Coluccia\corref{cor1}}
\ead{angelo.coluccia@unisalento.it}

\author[a]{Alessio Fascista}
\ead{alessio.fascista@unisalento.it}

\author[a]{Giuseppe Ricci}
\ead{giuseppe.ricci@unisalento.it}

\cortext[cor1]{Corresponding author}

\address[a]{DII, Universit\`a del Salento, Via Monteroni, 73100 Lecce, Italy.}

\begin{abstract}
A $k$-nearest neighbors (KNN) approach to the design of radar detectors is investigated. The idea is to start with either raw data or well-known radar receiver statistics as feature vector to be fed to the KNN decision rule. 
In the latter case, the probability of false alarm and probability of detection are characterized in closed-form; moreover, it is proved that the detector possesses the constant false alarm rate (CFAR) property and the relevant performance parameters are identified. 
Simulation examples are provided to illustrate the effectiveness of the proposed approach. 
\end{abstract}


\maketitle

\section{Introduction}\label{sec:intro}

The problem of radar detection has received significant attention over the past decades and is still an active field of research,  mostly employing statistical signal processing techniques based on hypothesis testing theory.
The classical detection framework has been set up by Kelly in his pioneering paper \cite{Kelly}. He derived  the
generalized likelihood ratio test (GLRT) based on 
the cell under test (CUT), also referred to as primary data, 
and a set of training or secondary data; such data are supposed to be independent and identically distributed random vectors, free of signal components, and  sharing with the CUT the statistical characteristics of the noise.
 In  \cite{Kelly89} the performance of such a detector is assessed when the actual
steering vector is not aligned with the nominal one.
Later, many  works have addressed the problem of enhancing either the selectivity or the robustness of  GLRT-based detectors to  mismatches. In particular, the adaptive matched filter (AMF) \cite{Kelly-Nitzberg} is a prominent example of robust detector, while the adaptive coherence estimator (ACE, also known as adaptive normalized 
matched filter)  \cite{Asymptotically,ACE} and Kelly's detector are selective receivers, i.e., they have excellent rejection capabilities of signals arriving from directions  different from the nominal one.
Other detectors try to explicitly take into account rejection capabilities at the design stage, as for instance those based on the adaptive beamformer orthogonal rejection test (ABORT)  \cite{Pulsone-Rader} or related ideas \cite{Fabrizio-Farina, W-ABORT, CR_SPL, Besson}.

Recently, the possibility to bring tools from machine learning to the radar context has started to be investigated. 
For instance, support vector machines (SVM) have been proposed for the design of radar detectors: in particular, a linear SVM approach has been adopted in \cite{Ball}, that is able to detect small SNR signals in situations where a  cell-averaging constant false alarm rate (CA-CFAR) 
scheme is unable. An SVM-based CFAR detector has also been proposed in \cite{LeiouWang}: it is robust in non-homogeneous environments
including multiple targets and clutter edge. Notably, SVM-based approaches are showing their effectiveness also for detection problems outside the radar domain, namely spectrum sensing \cite{SVM_JSAC,SVM_TSP,EUSIPCO2019}.
The $k$-nearest neighbors (KNN) approach has been used to detect radar signals in non-Gaussian noise \cite{CR_Boston2019}. Therein, modified Kelly's and ACE stastistics are the entries of the feature vectors. The proposed detector is not CFAR, but its probability of false alarm
($P_{fa}$) is not very sensitive to unknown disturbance (clutter plus thermal noise) statistics.
Finally, deep learning tools have been applied to detection, classification, and waveform generation for automotive radars
\cite{detection_radarcon19,classification_radarcon19,waveform_design_radarcon19}.

One of the main issues with the application of machine learning tools to radar detection is that it is generally very difficult to theoretically assess the performance of the resulting receivers.
In this paper, we make a step towards this direction by investigating the potential of a novel family of detectors based on the KNN approach. The latter is in fact one of the simplest machine learning algorithms for classification, since it basically performs computation of distances with respect to a training set, followed by a count-based decision rule (e.g., majority); by contrast, SVM for instance requires to solve numerically an optimization problem to obtain the decision rule, which hampers its theoretical analysis.
The contribution of the present paper is twofold. First, we statistically characterize the KNN detection procedure, providing general closed-form expressions for the probability of false alarm and probability of detection; second, we apply the proposed framework to the design and analysis of radar detectors based on different criteria (feature vectors). Numerical results are provided to illustrate the effectiveness of the detectors that can be obtained by the proposed approach. 

The remaining of the paper is organized as follows. Sec. \ref{sec:KNN} sets up the KNN-based detection problem and provides a general characterization of the achievable performance. Sec. \ref{sec:app_radar} is instead devoted to the design and analysis of radar detectors based on the proposed framework (derivation details reported in Appendix A). We conclude the paper in Sec. \ref{sec:conclusion}.

\section{KNN-based detectors}\label{sec:KNN}

\subsection{Problem formulation}

Let $\bm{o} \in O^n$ denote the $n$-dimensional observation vector containing the data collected over a given space $O^n$; they can be in general measurements obtained through a set of sensors, information stored in a database, etc. Data are typically mapped into a  lower dimensional space where some distinguishing characteristics of the observed phenomenon tend to emerge. This process is known as \emph{feature extraction} and can be defined as a function $\mathcal{F}: O^n  \rightarrow F^m$ that maps the observation vector $\bm{o} \in O^n$ into a \emph{feature vector} $\bm{x} \in F^m$, i.e., $\bm{x} = \mathcal{F}(\bm{o})$. As to $F^m$, it represents the feature space, assumed to be an $m$-dimensional Euclidean vector space ($\|\point\|$ 
will denote the Euclidean norm). In this context, a general two-class classification problem can be  formulated as the following binary hypothesis testing problem
\begin{equation}
\left\{
\begin{array}{ll}
H_{0}: & \bm{x}  \thicksim {\cal D}^0 \\ 
H_{1}: & \bm{x}  \thicksim {\cal D}^1 
\end{array}
\right. \label{eq::generictest}
\end{equation}
and consists in determining whether the feature vector $\bm{x}$ is distributed according to $\mathcal{D}^0$ ($H_0$ hypothesis) or $\mathcal{D}^1$ ($H_1$ hypothesis) that in general will depend on unknown parameters. Hereafter, we use superscripts ${ }^0$ and ${}^1$ to indicate data belonging to the $H_0$ and $H_1$ hypothesis, respectively. To assess the performance of classifiers for problem~\eqref{eq::generictest}, we will resort to the well-established type I and type II error metrics.  
More specifically, type I errors occur if the decision scheme decides for $H_1$ when $H_0$ is in force; if the $H_0$ hypothesis is simple, i.e., completely specified, the probability of a type I error is commonly denoted probability of false alarm ($P_{fa}$). On the other hand, type II errors occur whenever the decision scheme decides for $H_0$ when $H_1$ is true (missed detection) and its probability is equal to $1 - P_d$ with $P_d$ denoting the probability of detection, that is, the probability to decide for $H_1$ when $H_1$ holds true.  Radar detectors are commonly designed to guarantee $P_d$ values as close as possible to one for a preassigned value of the $P_{fa}$. 

To derive the KNN classifier, we need a training set containing representative examples of the observed data under both $H_0$ and $H_1$ hypotheses. Without loss of generality, we assume that $N_T$ independent observations of the raw data $\bm{o}$ under both $H_0$ and $H_1$  are available. Observation vectors $\bm{o}^0_i$ (under $H_0$) are used to construct 
the corresponding feature vectors $\bm{x}^0_i$, $i=1,\ldots,N_T$. Similarly, observation vectors $\bm{o}^1_i$ (under $H_1$) are used to construct 
the feature vectors $\bm{x}^1_i$, $i=1,\ldots,N_T$.
Accordingly,
we obtain the following training set
${\cal T}= {\cal T}^0 \cup  {\cal T}^1$
with
\begin{equation}
{\cal T}^0=
\left\{\bt_i^0= \begin{bmatrix}\bx_i^0 \\ 0\end{bmatrix} \in F^{(m+1) \times 1}, i=1, \ldots, N_T \right\} \label{setT0}
\end{equation}
where $\bx_i^0 = \left[ x_i^0[1] \ \cdots \   x_i^0[m]\right]^T $ ($^T$ denoting transposition), $x^0_i[j], j=1,\ldots,m$, denoting the $j$th feature obtained from the data under $H_0$, and similarly
\begin{equation}
{\cal T}^1=
\left\{\bt_i^1=\begin{bmatrix} \bx_i^1 \\ 1 \end{bmatrix} \in F^{(m+1) \times 1}, i=1, \ldots, N_T \right\} \label{setT1}
\end{equation}
with $\bx_i^1 = \left[ x_i^1[1] \ \cdots \   x_i^1[m]\right]^T$. Fig. \ref{fig:conceptschema} shows a two-dimensional example.

To implement the KNN-based decision rule, we associate to a given input data under test $\bm{o}$ the feature vector
$
\bx= \left[x[1] \ \cdots \ x[m]
\right] ^T
$; more precisely,
denoting by $\bt_i=\begin{bmatrix}
\bx_i \\ \ell_i \end{bmatrix}$, $i=1, \ldots, 2N_T$,
the elements of the training 
sample ${\cal T}$ 
(the ``label'' $\ell_i$ is either 0 or 1 depending 
on the fact that $\bt_i$ belongs to ${\cal T}^0$ or ${\cal T}^1$, respectively),
we compute the following statistic
$$
\overline{\ell}= \frac{1}{k} \sum_{\left\{ i: \ \bx_i \in N_k\left(\bx \right)\right\} } \ell_i
$$
with $N_k\left(\bx \right)$ the set of the $k$ vectors $\bx_i$s closest to the test vector $\bx$ according to the Euclidean norm $\|\point\|$ (``the $k$ nearest neighbors of $\bx$''), as highlighted in Fig. \ref{fig:conceptschema}.
\begin{figure}
\centering
\includegraphics[width=0.75\textwidth]{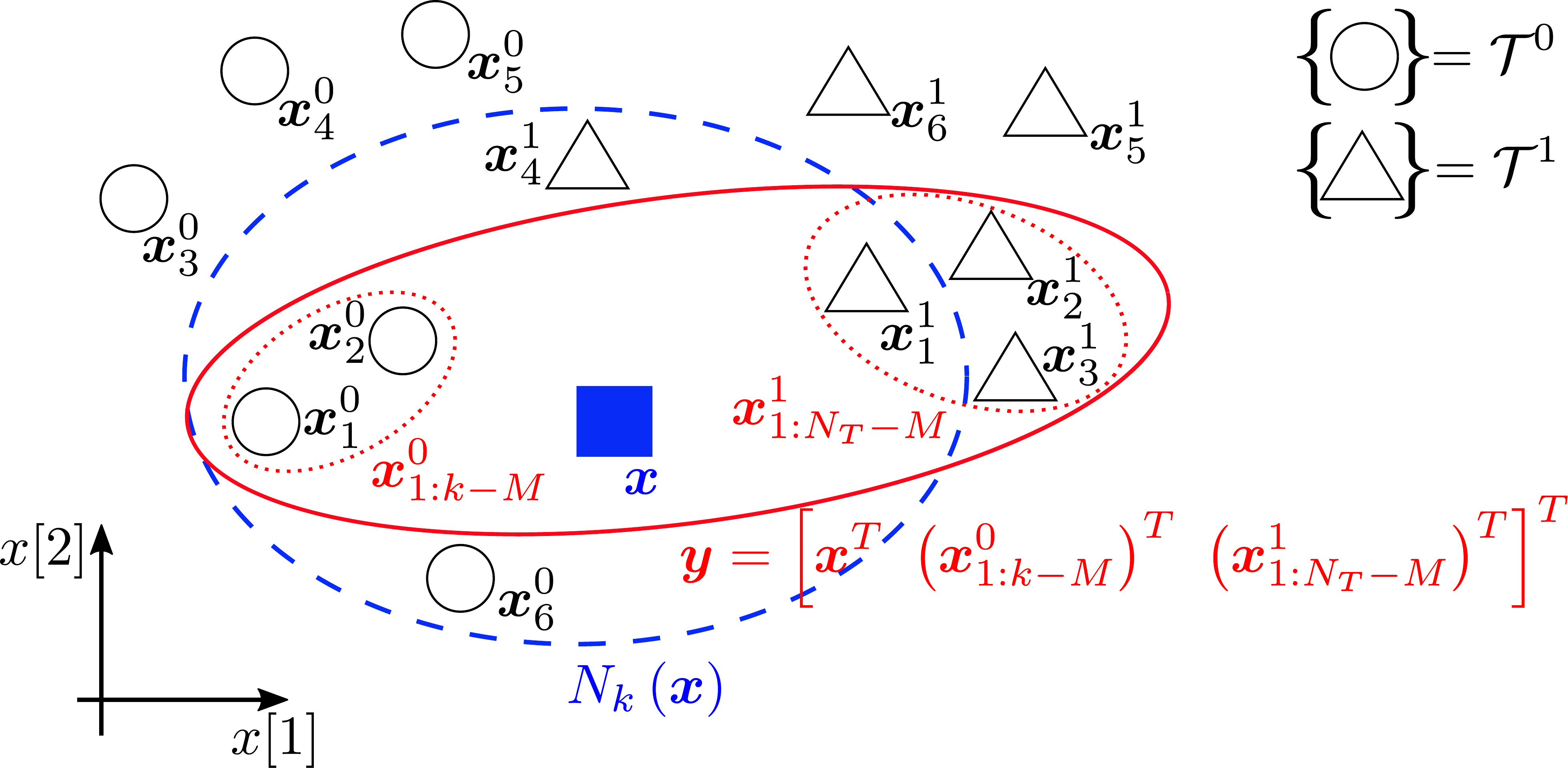} 
\caption{Graphical representation of a two-dimensional feature space, with $N_T=6$ training data (circles and triangles) for each hypothesis. The data under test is depicted as a (blue) square, and its $k=5$ nearest neighbors are enclosed in a dashed ellipse: in this example, for $M=3$ the resulting decision would be $H_0$. The solid (red) ellipse, conversely, groups the first $k-M$ elements of $\mathcal{T}^0$ and the first $N_T-M$ elements of $\mathcal{T}^1$, which for the specific case at hand correspond to an unfavorable realization of the event ${\cal B}_{j_1: j_{k-M}; i_1 : i_{N_T-M}}$.
}
\label{fig:conceptschema}
\end{figure}
Finally, $H_0$ or $H_1$ is selected
according to the decision rule
$$
\overline{\ell} \test T
$$
where $T$ is a chosen detection threshold.
Notice that, due to the fact that the $\ell_i$ are binary digits, the test can be equivalently re-written as
\be
\left\{
\begin{array}{ccl}
\mbox{choose } &H_0: & \mbox{if } \# \ \bx_i^1\in N_k\left(\bx \right)  \leq M
\\
\text{choose } &H_1: & \mbox{otherwise}
\end{array}
\right.
\label{eq:rule2}
\ee
with $M$ the greatest integer such that $T \geq M/k$ and $\#$ stands for ``the number of''.
Notice also that $\overline{\ell}$ is a discrete random variable and, hence, different values of $T$ 
do not necessarily correspond to different values of $P_{fa}$ (for a deterministic test).

\subsection{Performance assessment of KNN detectors}

In this section, we provide closed-form analytical formulas for the $P_{fa}$ and the $P_d$ of the proposed KNN-based approach, which are useful to predict the achievable classification performance and offer some insights to interpret the classification process, as shown later in Sec. \ref{sec:app_radar}.

For future reference $E[\point]$ denotes the expectation operator, 
$$
\bm{x}^0_{1:k-M} = \left[(\bx_{1}^0)^T \ \cdots \ (\bx_{k-M}^0)^T \right]^T, \quad \bm{x}^1_{1:N_T-M} = \left[(\bx_{1}^1)^T \ \cdots \ (\bx_{N_T-M}^1)^T\right]^T,$$
\be
\bm{y} = 
\left[\bm{x}^T \ \left(\bm{x}^{0}_{1:k-M}\right)^T \ \left(\bm{x}^{1}_{1:N_T-M} \right)^T \right]^T 
\label{eq:defy}
\ee
For the case of a generic vector of features $\bx$, we have the following result.
\begin{myteo}
Let $\bx \in F^m$ be an  $m$-dimensional vector of features, constructed from the data under test $\bm{o} \in O^n$, and ${\cal T} = {\cal T}^0 \cup {\cal T}^1$ represent a training set containing $N_T$ independent realizations  under both $H_0$ and $H_1$ hypotheses, denoted as $\bx_i^0$ and $\bx_i^1$, respectively, $i=1,\ldots,N_T$ (ref. eqs. (\ref{setT0})-(\ref{setT1})). The following expression holds true for the probability that the KNN-based decision statistic in eq.~\eqref{eq:rule2} exceeds the threshold $T$:
\begin{eqnarray*}
P(\overline{\ell} > T) &=&
1-
\left(
\begin{array}{c}
N_T \\ k-M
\end{array}
\right)
\left(
\begin{array}{c}
N_T \\ N_T-M
\end{array}
\right)
\\ &\times&
E_{ \by 
}\left[ I_{\cal Y}(\by)
\left( p_0 \left( \bx, \bm{x}^0_{1:k-M} \right) \right)^{N_T-k+M}
\left( p_1 \left( \bx, \bm{x}^1_{1:N_T-M} \right) \right)^{M}
\right]
\end{eqnarray*}
with
$$
p_0 \left( \bx, \bm{x}^0_{1:k-M} \right) \!= \! P \left(  \left\{ \|\bx^0-\bx\| \geq  \max_{r \in \{1, \ldots, k-M\}} \|\bx_{r}^0-\bx\|\right\} \Big| \bx, \bm{x}^0_{1:k-M} \right)
$$
$$
p_1  \!\!\left( \bx, \bm{x}^1_{1:N_T-M} \right) \!=\! P \! \left(  \!\left\{ \|\bx^1-\bx\| \leq \!\!\!\! \min_{r \in \{1, \ldots, N_T-M\}} \!\! \|\bx_{r}^1-\bx\| \right\} \Big| \bx, \bm{x}^1_{1:N_T-M} \right)
$$
where $\bx^0$ and $\bx^1$ are generic random variables distributed according to $\mathcal{D}^0$ or $\mathcal{D}^1$, respectively (ref. eq.~(\ref{eq::generictest})), and $I_{\cal Y}(\by)$ is the indicator function of the set ${\cal Y}$  introduced to constrain the above probability to be nonzero only if $\bm{y} \in {\cal Y}$, 
\begin{align*}
 {\cal Y} = \{\by: 
&\|\bx_{1}^0 - \bx\|  \leq \|\bx_{1}^1-\bx\|, \ldots,  
\|\bx_{1}^0-\bx\| \leq  \|\bx_{N_T-M}^1-\bx\|, \ldots, 
\\
&\|\bx_{k-M}^0-\bx\| \leq \|\bx_{1}^1-\bx\| , \ldots,   
\|\bx_{k-M}^0-\bx\| \leq \|\bx_{N_T-M}^1-\bx\|
\}.
\end{align*}
\end{myteo}
\begin{proof}
The proof relies on a proper decomposition of the event
$$
\left\{\# \ \bx_i^1\in N_k\left(\bx \right)  \leq M \right\},
$$
corresponding to outcomes of the underlying random experiment that lead the KNN classifier to decide for $H_0$.
Obviously,  $P(\overline{\ell} > T)$ is the probability of the complementary of the above event.
In formulas, the KNN classifier decides for $H_0$ only if the outcomes of the random experiment belong to
\be
\bigcup_{\stackrel{j_1, \ldots, j_{k-M} \in {\cal N}_T}{i_1, \ldots, i_{N_T-M} \in {\cal N}_T}}
{\cal B}^{01}_{j_1: j_{k-M}; i_1: i_{N_T-M}}
\label{union_of_nonexclusive_events}
\ee
with
\begin{align*}
{\cal B}^{01}_{j_1: j_h; i_1: i_{N_T-(k-h)}} &=
\left\{ 
\|\bx_{j_1}^0-\bx\| \leq \|\bx_{i_1}^1-\bx\|, \ldots,  
\|\bx_{j_1}^0-\bx\| \leq \|\bx_{i_{N_T-(k-h)}}^1-\bx\|,  
\right.
\\&
\left.
\ldots,\|\bx_{j_h}^0-\bx\| \leq \|\bx_{i_1}^1-\bx\| , \ldots,   
\|\bx_{j_h}^0-\bx\| \leq \|\bx_{i_{N_T-(k-h)}}^1-\bx\|
\right\}
\end{align*}
and ${\cal N}_T=\{1, \ldots, N_T\}$.
In fact,
the event
${\cal B}^{01}_{j_1: j_h; i_1: i_{N_T-(k-h)}}$, with $h \leq k$, 
occurs when, given a subset of $h$ feature vectors belonging to ${\cal T}^0$,  indexed by $j_1, \ldots, j_h$, and a subset of $N_T - (k-h)$ feature vectors belonging to ${\cal T}^1$, indexed by $i_1, \ldots, i_{N_T-(k-h)}$, each element in the first subset is closer to the test vector $\bx$ than any other element in the second subset.
Thus, the event defined by eq. (\ref{union_of_nonexclusive_events}), where
$h=k-M$, 
is tantamount to imposing that at most $M$ feature vectors 
of ${\cal T}^1$ belongs to $N_k\left(\bx \right)$.
However, the elements of the (finite) union
in eq. (\ref{union_of_nonexclusive_events})
are not mutually exclusive events. For this reason, we introduce the additional events
$$
{\cal B}^{00}_{j_1: j_h}= 
\left\{
\max_{s \in  \{j_1, \ldots, j_{h}\}} \|\bx_{s}^0-\bx\| \leq \min_{r \in {\cal N}_T \setminus 
\{ j_1, \ldots, j_h \}} \|\bx_{r}^0-\bx\|
\right\}
$$
and
$$
{\cal B}^{11}_{i_1: i_{N_T-(k-h)}}= 
\left\{
\min_{s \in  \{i_1, \ldots, i_{N_T-(k-h)}\}} \|\bx_{s}^1-\bx\| \geq \max_{r \in {\cal N}_T \setminus 
\{i_1, \ldots, i_{N_T-(k-h)}\}} \|\bx_{r}^1-\bx\|
\right\}.
$$
 As to ${\cal B}^{00}_{j_1: j_h}$, it denotes the event that the $h$ feature vectors belonging to ${\cal T}^0$ and indexed by $j_1, \ldots, j_h$, are the closest to $\bx$ among all the vectors in the training set ${\cal T}^0$. Similarly, ${\cal B}^{11}_{i_1: i_{N_T-(k-h)}}$ occurs when the $N_T-(k-h)$ feature vectors, 
belonging to ${\cal T}^1$ and indexed by $i_1, \ldots, i_{N_T-(k-h)}$,
 are the farthest from $\bx$ among all the vectors in the training set ${\cal T}^1$. Using the above definitions, 
we define the event
$$
{\cal B}_{j_1: j_{k-M}; i_1 : i_{N_T-M}} = {\cal B}^{01}_{j_1: j_{k-M}; i_1 : i_{N_T-M}} \cap
{\cal B}^{00}_{j_1: j_{k-M}} \cap {\cal B}^{11}_{i_1: i_{N_T-M}}
$$
and re-write eq. (\ref{union_of_nonexclusive_events})
as the union of mutually exclusive events,
i.e.,
\be
\bigcup_{\stackrel{j_1, \ldots, j_{k-M} \in {\cal N}_T}{i_1, \ldots, i_{N_T-M} \in {\cal N}_T}}
{\cal B}_{j_1: j_{k-M}; i_1: i_{N_T-M}}.
\label{union_of_exclusive_events}
\ee
To give a concrete example, in Fig.~\ref{fig:conceptschema} we report a possible realization of the elements in the feature space that would lead the KNN detector to decide for $H_0$. In general, $P(\overline{\ell} > T)$ can be obtained as the (unconditional) complementary probability of the union over all possible combinations that produce the event ${\cal B}_{j_1: j_{k-M}; i_1 : i_{N_T-M}}$, namely as
\begin{equation}
P(\overline{\ell} > T) = 1-P \left( \bigcup_{\stackrel{j_1, \ldots, j_{k-M} \in {\cal N}_T}{i_1, \ldots, i_{N_T-M} \in {\cal N}_T}} \!\!\!\!\!
{\cal B}_{j_1: j_{k-M}; i_1 : i_{N_T-M}}\right).
\label{eq:PFA_KNN}
\end{equation}
To compute the right-hand side of the above formula, we can simply count all the possible combinations of $k-M$ and $N_T-M$ feature vectors from ${\cal T}^0$ and ${\cal T}^1$, respectively, and multiply by the probability of the event
${\cal B}_{1:k-M; 1: N_T-M} = {\cal B}^{01}_{1: k-M; 1 : N_T-M} \cap
{\cal B}^{00}_{1: k-M} \cap {\cal B}^{11}_{1: N_T-M}$
(namely, by choosing the first $k-M$ and $N_T-M$ indexes, which can be of course in any order with respect to the distance from $\bx$, as highlighted in Fig.~\ref{fig:conceptschema} by red dotted ellipses).  Then, eq.~(\ref{eq:PFA_KNN}) can be written as
$$
P(\overline{\ell} > T)= 1-
\left(
\begin{array}{c}
N_T \\ k-M
\end{array}
\right)
\left(
\begin{array}{c}
N_T \\ N_T-M
\end{array}
\right)
P \left( 
{\cal B}_{1:k-M; 1: N_T-M}
\right).
$$
The thesis follows by observing that, conditioned on $\bm{y}$, defined in eq.~(\ref{eq:defy}), the above joint probability is nonzero only if $\bm{y} \in \cal{Y}$ and can be factorized as
\begin{align*}
P & \left( 
{\cal B}_{1:k-M; 1: N_T-M}
| \by
\right) \\
&=
\left[ 
P \left(  \left\{ \|\bx^0-\bx\| \geq  \max_{r \in \{1, \ldots, k-M\}} \|\bx_{r}^0-\bx\|\right\} | \bx, \bm{x}^0_{1:k-M} \right)\right]^{N_T-k+M}
\\ &\times
\left[ 
P \left(  \left\{ \|\bx^1-\bx\| \leq \min_{r \in \{1, \ldots, N_T-M\}} \|\bx_{r}^1-\bx\| \right\}| \bx, \bm{x}^1_{1:N_T-M} \right)\right]^{M}
 \\&=
\left( p_0 \left( \bx, \bm{x}^0_{1:k-M} \right) \right)^{N_T-k+M}
\times
\left(p_1 \left( \bx, \bm{x}^1_{1:N_T-M} \right) \right)^{M}
\end{align*}
where $\bx^0$ is any chosen element of $\mathcal{T}^0$ not included in $\bx^0_{1:k-M}$ (i.e., $\bx^0 \in \{\bx^0_{k-M+1},\ldots,\bx^0_{N_T} \}$); analogously,  $\bx^1$ is any chosen element of $\mathcal{T}^1$ not included in $\bx^1_{1:N_T-M}$ (i.e., $\bx^1 \in \{\bx^1_{N_T-M+1},\ldots,\bx^1_{N_T} \}$). Thus, $\bx^0$ and $\bx^1$ are generic random vectors distributed according to $\mathcal{D}^0$ and $\mathcal{D}^1$, respectively (ref. eq.~(\ref{eq::generictest})).


\end{proof}

Proposition 1 allows one to analytically express $P_d$ and $P_{fa}$ in terms of the two probabilities $p_0$ and $p_1$, which are related to  elementary events in the feature space. This result is fully general and does not depend on the distribution of the data (features). We will show in the sequel that, given a specific feature vector, such a tool can be used to prove the CFAR property of the resulting detector, and to identify its relevant performance parameters.

Before proceeding, we provide a simple clarifying example. Suppose that the training data and the feature vector under test are complex normal  with an expected value equal to $\bme_i$ under $H_i$, $i=0,1$, and a scalar covariance matrix. It follows that the hypothesis testing problem \eqref{eq::generictest} becomes
\begin{equation}
\left\{
\begin{array}{ll}
H_{0}: & \bm{x}  \thicksim {\cal CN}_m\left(\bme_0,  \sigma^2 \bI_m \right)\\ 
H_{1}: & \bm{x}  \thicksim {\cal CN}_m\left(\bme_1,  \sigma^2 \bI_m \right)
\end{array}
\right. \label{eq::example}
\end{equation}
where $\sigma^2>0$ and $\bI_m$ is the $m \times m$ identity matrix. 
Since, conditioned on $\bx$, we have that
\be
\bx^0-\bx \thicksim {\cal CN}_m\left(\bme_0-\bx,  \sigma^2 \bI_m \right)
\label{def:RV1}
\ee
and
\be
\bx^1-\bx \thicksim {\cal CN}_m\left(\bme_1-\bx,  \sigma^2 \bI_m \right)
\label{def:RV2},
\ee
it turns out that the norm squared of the random variables (RVs) \eqref{def:RV1} and \eqref{def:RV2},
normalized by $\sigma^2$,
are
complex noncentral chi-square RVs with $m$ degrees of freedom and noncentrality parameters $\frac{\| \bme_0-\bx \|}{\sigma}$ and $\frac{\|\bme_1-\bx \|}{\sigma}$, respectively.
In symbols,
we write
$
\frac{1}{\sigma^2} \|\bx^0-\bx \|^2 \thicksim  {\cal C}_{{\cal \chi}^2_{m}}\left(\frac{\| \bme_0-\bx \|}{\sigma} \right)
$
and
$
\frac{1}{\sigma^2} \| \bx^1-\bx \|^2 \thicksim {\cal C}_{{\cal \chi}^2_{m}}\left( \frac{\|\bme_1-\bx \|}{\sigma} \right).
$
Thus, $p_0 \left( \bx, \bm{x}^0_{1:k-M} \right) $ and 
$p_1 \left( \bx, \bm{x}^1_{1:N_T-M} \right)$ can be  computed 
in terms of the corresponding cumulative distribution functions (CDFs). 


\section{Application of the KNN approach to adaptive radar detection}\label{sec:app_radar}

In this section, we demonstrate that KNN-based decision schemes can be fruitfully applied to design novel radar detectors. We recall that the well-known  problem of detecting the possible presence of a coherent return from a given CUT in range, doppler, and azimuth, is classically formulated as the following
hypothesis testing problem:
\begin{equation}
\left\{
\begin{array}{ll}
H_{0}: & \bz =  \bn \\
H_{1}: & \bz = \alpha \bv + \bn
\end{array} 
\right.\label{eq:test1ord}
\end{equation}
where
$\bz \in \C^{N \times 1}$, $\bn \in \C^{N \times 1}$, and $\bv \in \C^{N \times 1}$ denote 
the received vector, the corresponding noise term, and 
the known steering vector of the useful target echo.
The noise term is commonly modeled according to the complex normal distribution
with zero mean and unknown (Hermitian) positive definite matrix $\bC$,  i.e., $\bn \sim {\cal CN}_N (\bzero, \bC)$.
Modeling $\alpha \in \C$ has an unknown deterministic parameter
returns a complex normal distribution for $\bz$ under both hypotheses; the non-zero mean of 
the received vector under $H_1$ makes it possible to discriminate between the two hypotheses, namely by resorting to the GLRT.

The above classical approach led to a number of well-known receivers following the pioneering paper by Kelly \cite{Kelly}, as recalled in Sec. \ref{sec:intro}; they exploit (in addition to the primary received signal $\bz$) a set of secondary data
$\bor_1, \ldots, \bor_{K_S}$, independent of $\bz$, free of signal components, and  sharing with the CUT the statistical characteristics of the noise.
As already observed, in this paper we want to investigate the potential of choosing between $H_0$ and $H_1$ based on a KNN classifier. For the specific radar detection problem at hand, we consider as input data the vector obtained by stacking both primary and secondary data, namely $\bm{o} = [\bz^T \ \bor^T_1 \cdots \bor^T_{K_S}]^T$, and define $\bS$ as $K_S$ times the sample covariance matrix based on secondary data, namely
$$
\bS=\sum_{i=1}^{K_S} \bor_i \bor_i^H
$$
with $^H$ denoting the complex conjugate transpose. 
In the following we develop the KNN approach for different choices of the feature vector.

\subsection{First approach: a solution based on a raw data}\label{sec:first}

We propose to use as feature vector $\bx$ the ``whitened data'' under test, i.e., 
\be
\bm{x}=\bS^{-1/2} \bz.
\label{eq:raw_data}
\ee
Moreover, in radar detection problems it is customary to assume that signals backscattered by moving targets follow a statistical model as, for instance, the one provided in the hypotheses test \eqref{eq:test1ord}. The availability of a model for the data introduces significant advantages in the KNN training procedure, allowing the whole training set $\cal{T}$ to be constructed artificially, without requiring any preliminary collection phase, as instead typical in the majority of machine learning problems. For the specific case at hand, we can draw $N_T$ independent realizations of the observation vector $\bm{o} = [\bz^T \ \bor^T_1 \cdots \bor^T_{K_S}]^T$ under both $H_0$ and $H_1$ hypotheses. 
More precisely, the training data under $H_0$, i.e., $\bx_i^0, i=1,\ldots,N_T,$ are generated assuming  $\bz, \bor_1,\ldots,\bor_{K_S} \thicksim {\cal CN}_N (\bzero, \bC)$
with $\bm{C}$ a preassigned (i.e., a design value) covariance matrix.
Similarly, the training data under $H_1$, i.e., $\bx_i^1, i=1,\ldots,N_T,$ are generated according to  $\bz \thicksim {\cal CN}_N (\alpha \bv, \bC)$, $\bor_1,\ldots,\bor_{K_S}  \thicksim {\cal CN}_N (\bzero, \bC)$,
and based on a preassigned value of the nominal signal-to-noise ratio (SNR), defined as
\begin{equation}
\mbox{SNR} = |\alpha|^2 \bm{v}^H \bm{C}^{-1}\bm{v}. \label{nominalSNR}
\end{equation}

The performance assessment is conducted  in comparison to natural references, namely Kelly's detector, AMF, and ACE, given by
$$
t_{\text{\tiny{Kelly}}} (\bz,\bS) =
 \frac{|\bz^{H} \bS^{-1} \bv |^2}{\bv^{H} \bS^{-1} \bv \  \left( 1+ \bz^{H} \bS^{-1} \bz \right) },
$$
$$
t_{\text{\tiny{AMF}}} (\bz,\bS) =
 \frac{|\bz^{H} \bS^{-1} \bv |^2}{\bv^{H} \bS^{-1} \bv}, \quad
t_{\text{\tiny{ACE}}} (\bz,\bS) =
 \frac{|\bz^{H} \bS^{-1} \bv |^2}{\bv^{H} \bS^{-1} \bv \  \bz^{H} \bS^{-1} \bz },
$$
respectively,
with
$(\cdot)^{-1}$ the inverse of the matrix argument and $|\cdot|$ the modulus of the argument variable, respectively.
We estimate the $P_d$s 
and the $P_{fa}$ over $10^3$ and $10^5$ trials, respectively.
For simulation purposes,
we assume that the nominal (but also the actual) steering vector
has the  form 
$
\bv=[ 1 \ e^{j 2 \pi \nu_d} \cdots e^{j 2 \pi \nu_d (N-1)}]^T
$, 
thus focusing on a single antenna processing,
where $\nu_d$ is the normalized Doppler frequency shift\footnote{$\nu_d=f_dT_{\text{\tiny PRT}}$ with $f_d$ the target Doppler frequency shift and $T_{\text{\tiny PRT}}$ the pulse repetition time of the radar.}.  
The noise covariance matrix is Gaussian-shaped with  one-lag correlation coefficient 0.95.

We set $N=8$ and $K_S=16$ and generate $N_T=10^3$ training data for each hypothesis, so obtaining ${\cal T}_0$ and ${\cal T}_1$; for ${\cal T}_1$ we assume design parameters $f_d=0.08$ and SNR = 12 dB. The KNN is implemented with $k=50$ (using the Euclidean distance as metric) and the threshold is set to $T=1/2$, i.e., the algorithm chooses $H_1$ if at least 26 out of 50 closest data belong to ${\cal T}_1$.

\begin{figure}
\centering
\includegraphics[width=0.7\textwidth]{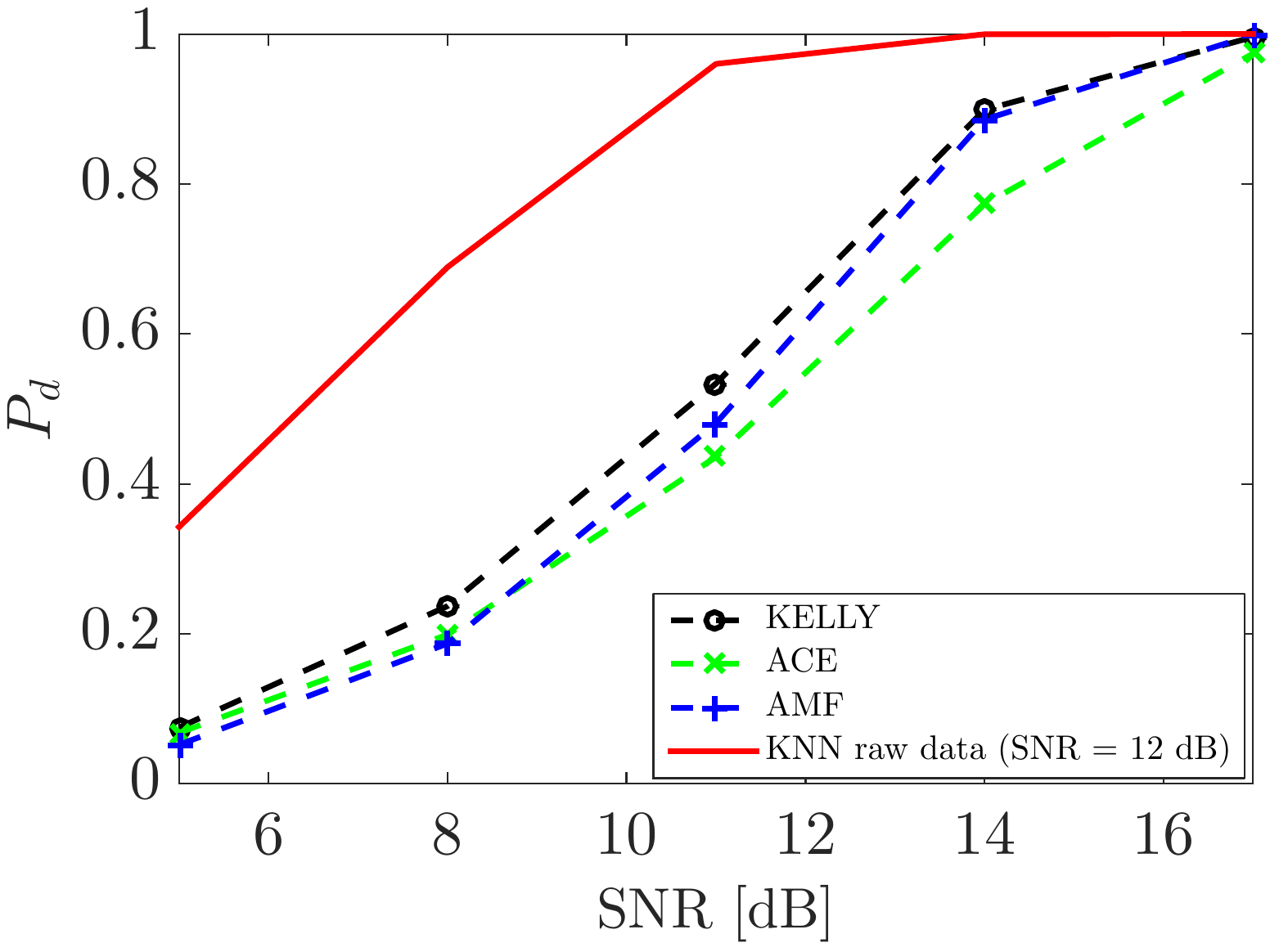} 
\caption{$P_d$ vs SNR under matched conditions for KNN using raw data $\bm{S}^{-1/2}\bm{z}$.}
\label{Fig3}
\end{figure}

\begin{figure}
\centering
\includegraphics[width=0.7\textwidth]{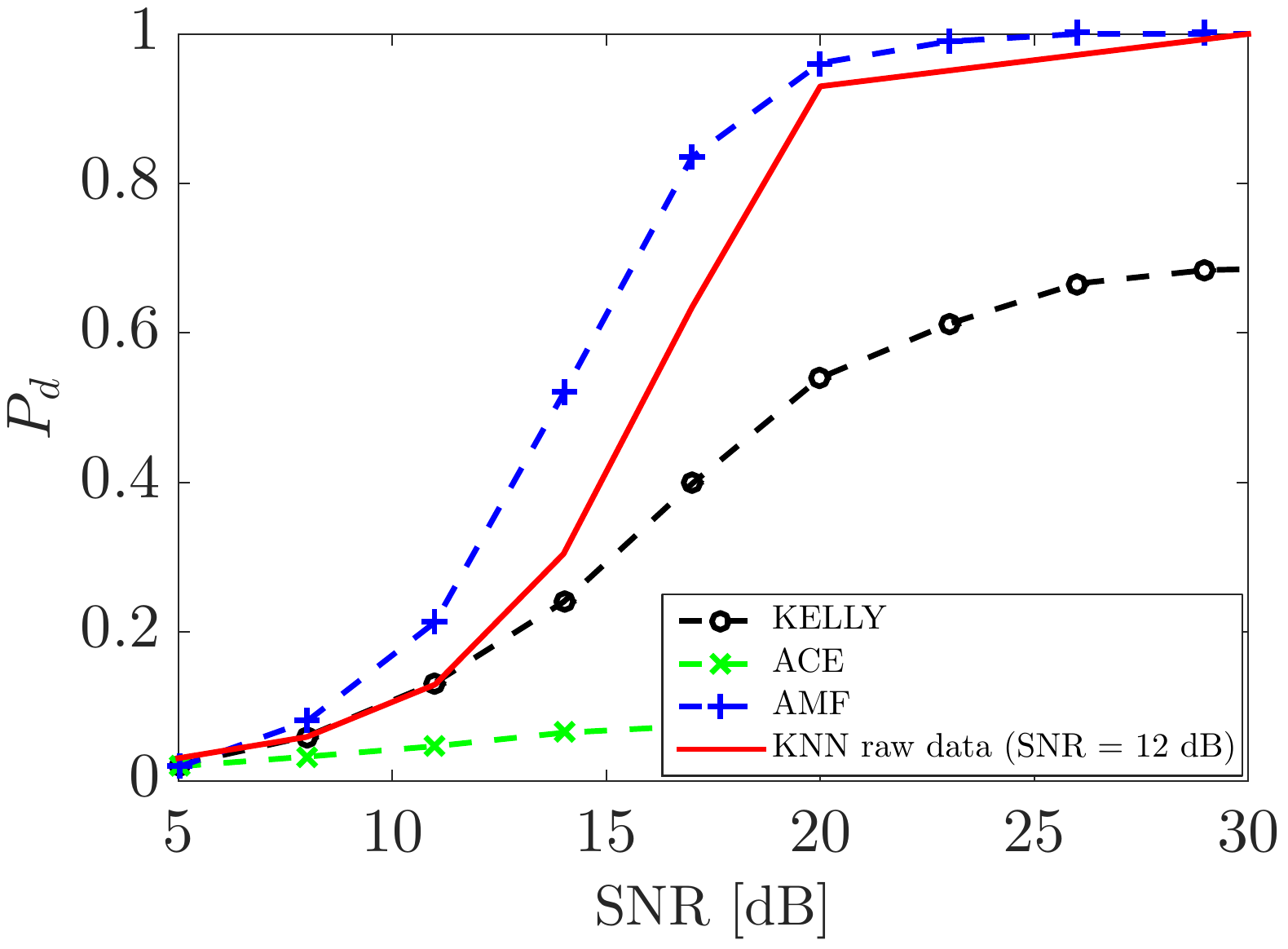} 
\caption{$P_d$ vs SNR under mismatched conditions with $\cos^2 \theta = 0.50$ for KNN using raw data $\bm{S}^{-1/2}\bm{z}$.}
\label{Fig4}
\end{figure}

First consider the case that the actual Doppler frequency is perfectly matched to the design value. 
The performance of the KNN detector is compared to AMF, Kelly's detector, and ACE in Figure~\ref{Fig3}. With the chosen parameters, we obtained $P_{fa}=0.0048$ for the KNN detector, and the threshold setting for AMF, Kelly's detector, and ACE has been performed to guarantee the same $P_{fa}$.
It is apparent that the proposed approach can achieve a very high $P_d$ even for reduced SNR values, that is, it is much more powerful than the competitors, including Kelly's detector which requires more than 4 dB of SNR to achieve the same performance. The results under mismatched conditions are also interesting.
In particular, Fig.~\ref{Fig4} shows the performance for the case in which the nominal Doppler frequency is perturbed by an additive term $0.4/N$, corresponding to a cosine squared $\cos^2 \theta = 0.50$ of the angle between the nominal steering vector $\bv$ and the mismatched one $\bp$, i.e.,
 $$
\cos^2 \theta=
\frac{|\bp^H \bC^{-1} \bv|^2}{\bp^H \bC^{-1} \bp \ \bv^H \bC^{-1} \bv}.
$$
The results show that, although the proposed KNN-based detector experiences a significant $P_d$ loss compared to the matched case, it is anyway more powerful than Kelly's detector and close to the AMF, but with the great advantage of strong detection capabilities under matched conditions. The price to pay for this excellent performance is the loss of the CFAR property. To investigate how changes in the noise distribution affect the resulting $P_{fa}$, a simulation analysis assuming a different value for the one-lag correlation coefficient, now set to 0.5, is conducted. For this setup, the resulting $P_{fa} = 0.0062$, meaning that the proposed detector is quite insensitive to changes in the noise statistics, despite it does not strictly possess the CFAR property. In next section, we will show that such a property can be recovered by using  a different feature vector; moreover, we will show that the same power of Kelly's detector can be obtained, but with a  level of robustness or selectivity that can be controlled by tuning some design parameters.

\subsection{Second approach: a CFAR solution based on known statistics}

To come up with a CFAR detector,
we propose to use a vector of features $\bm{x}$  obtained by stacking  (compressed) statistics of some well-known radar detectors (among them we cite here AMF, Kelly's detector, ACE, Energy detector, W-ABORT, etc.). 
It is very important to notice that all the above statistics can be decomposed in terms of common statistics.
To be more definite, consider the AMF and  Kelly's detector given in the previous section;
they can be decomposed as
\cite{BOR-MorganClaypool}
\begin{equation}
t_{\text{\tiny{AMF}}} 
=
\frac{\tilde{t}}{\beta}, \qquad
t_{\text{\tiny{Kelly}}}=\frac{\tilde{t}}{1+\tilde{t}} \label{eq:AMF}
\end{equation}
with
\begin{equation}
\tilde{t}=
\frac{t_{\text{\tiny{Kelly}}}}{1-t_{\text{\tiny{Kelly}}}} \label{eq:Kelly}
\end{equation}
and
\begin{equation}
\beta=\frac{1}{1+ \bz^{H} \bS^{-1} \bz - \frac{|\bz^{H} \bS^{-1} \bv |^2}{\bv^{H} \bS^{-1} \bv}}. \label{eq:beta}
\end{equation}
More generally, all  the above listed statistics share a common dependency on both $\tilde{t}$ and $\beta$, which suggests that they can be conveniently stacked to construct a feature vector $\bx$ having the following structure 
\begin{equation}
\bm{x} = \left[d_1\tilde{t}b[1] \ d_2\tilde{t}b[2] \ \cdots \ d_m\tilde{t}b[m] \right]^T
\label{eq::vect_struct}
\end{equation}
with
\be 
b[j] = f_j(\beta), \quad j=1,\ldots,m, \label{eq:btransf}
\ee
denoting an arbitrary (nonlinear) function of the $\beta$ statistic and $\bD=\mbox{diag}(d_1, \ldots, d_m)$ 
an arbitrary (nonnegative) diagonal matrix  that introduces a set of additional degrees of freedom,  better explained in the following. 

For the specific choice of $\bm{x}$ given in \eqref{eq::vect_struct}, we have the following results.

\begin{myteo}
Let $\bx = \left[d_1\tilde{t}b[1] \ d_2\tilde{t}b[2] \ \cdots \ d_m\tilde{t}b[m] \right]^T$ be an $m$-dimensional feature vector constructed from the data under test $\bm{o} = [\bz^T \ \bor^T_1 \cdots \bor^T_{K}]^T$ and ${\cal T}$ a training set containing $2N_T$ independent realizations ($N_T$ under $H_0$ and $N_T$ under $H_1$). The probabilities $p_0 \left( \bx, \bm{x}^0_{1:k-M} \right)$ and $p_1 \left( \bx, \bm{x}^1_{1:N_T-M} \right)$ given in Proposition 1 and involved in the computation of $P_{fa}$ and $P_d$ can be expressed in closed-form, as shown in Appendix A.
Based on them, it turns out that the $P_{fa}$ depends upon the $\mbox{SNR}$ \eqref{nominalSNR} used to generate the training data, but is otherwise independent of the actual covariance matrix $\bC$, that is, the detector possesses the constant false alarm rate (CFAR) property.
In addition, $P_d$ depends only upon the $\mbox{SNR}$ \eqref{nominalSNR} used to generate the training data, as well as on
$$
\mbox{SNR}_p=
|\alpha|^2 \bp^H \bC^{-1} \bp \quad \text{ and } \quad
\cos^2 \theta=
\frac{|\bp^H \bC^{-1} \bv|^2}{\bp^H \bC^{-1} \bp \ \bv^H \bC^{-1} \bv}.
$$
\end{myteo}
\begin{proof}
See Appendix A.
\end{proof}

Proposition 2 shows that, although the proposed design is quite different from the traditional approach to radar detection, and specifically that based on the GLRT or other statistical hypothesis testing tools,  it is still possible to obtain a receiver with desirable properties: in particular, the CFAR property allows one to obtain the same $P_{fa}$ irrespective of the unknown noise statistics, which is very important in practical applications. Moreover, the performance in terms of $P_d$ depends on the classical parameters SNR and cosine squared of the angle between the nominal and actual steering vector, as in most well-known CFAR detectors (Kelly's detector, AMF, ACE, etc.).

An additional advantage of the proposed KNN-based approach is that, by acting on the choice of the statistics in the feature vector and on their relative weights given by the diagonal matrix $\bD$, it is possible to obtain  a detector that is either more robust or more selective than Kelly's detector. In particular, we illustrate in the following the performance of a detector whose feature vector is composed by the Kelly's detector and AMF, i.e.,
\be
\bx = \left[d_1\tilde{t} \ \ d_2 \frac{\tilde{t}}{\beta} \right]^T \label{eq:CFARfeat}
\ee
where eqs. \eqref{eq:AMF}, \eqref{eq:Kelly}, and \eqref{eq:beta} have been used.

In this case, we aim at obtaining a detector that is more robust than Kelly's detector, but with limited $P_d$ loss under matched conditions. 
Without loss of generality, we set $d_1=1$ so leaving $d_2$ as the sole tunable parameter. Simulations are conducted as for the previous case of Sec. \ref{sec:first}, but for $N=16$ and $K_S=32$; moreover, we set $d_2=0.7$.

\begin{figure}
\centering
\includegraphics[width=0.7\textwidth]{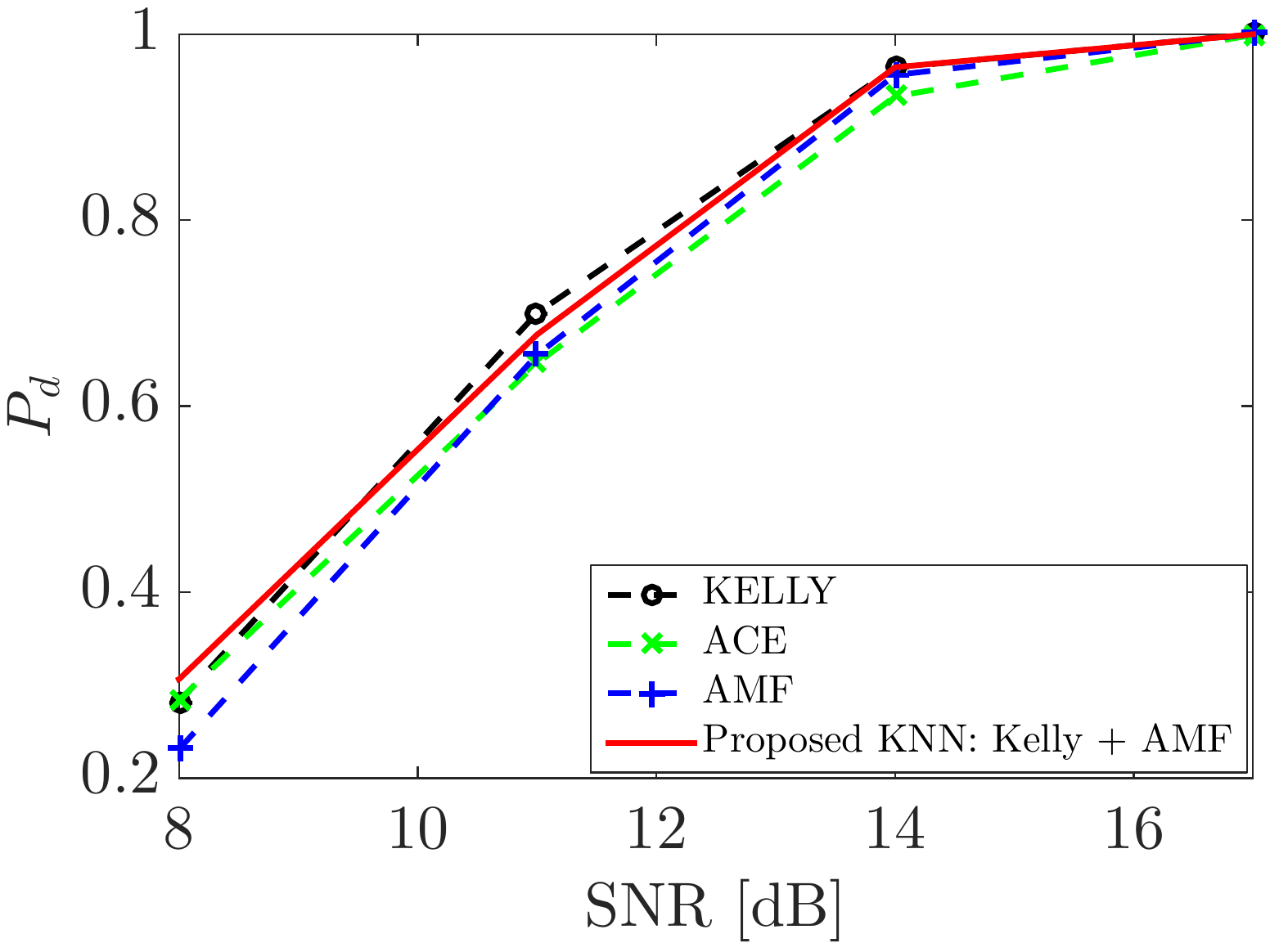} 
\caption{$P_d$ vs SNR under matched conditions for
a KNN fed by eq. (\ref{eq:CFARfeat}).}
\label{Fig1}
\end{figure}

\begin{figure}
\centering
\includegraphics[width=0.7\textwidth]{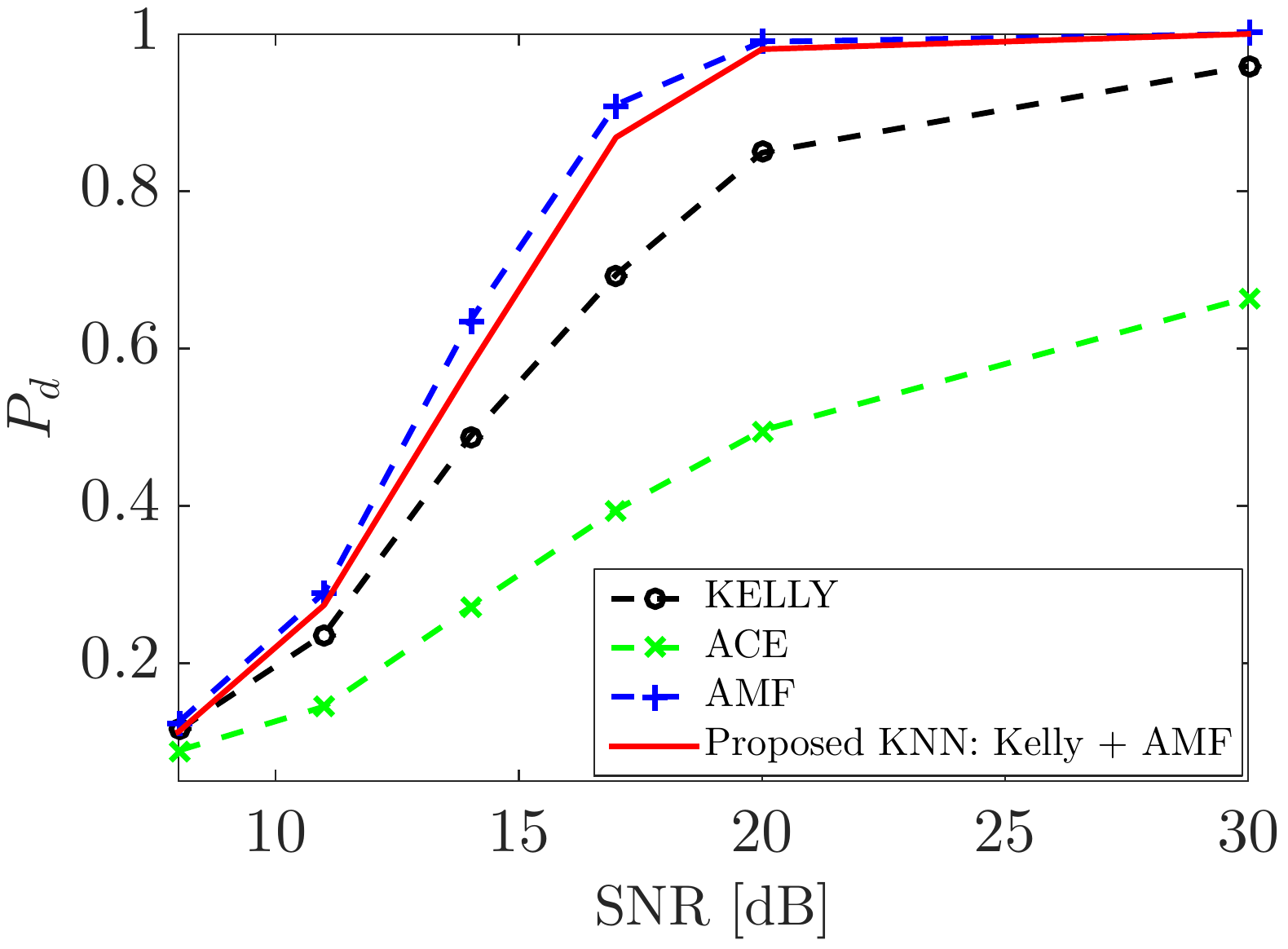} 
\caption{$P_d$ vs SNR under mismatched conditions with $\cos^2 \theta = 0.46$, for a KNN fed by eq. (\ref{eq:CFARfeat}).}
\label{Fig2}
\end{figure}

Fig. \ref{Fig1} shows the results under matched conditions.
It is apparent that the KNN detector has practically the same $P_d$ of Kelly's detector, followed by the AMF (which experiences some loss, especially at low SNR) and ACE (which experiences a loss of  about 3 dB at $P_d=0.9$).
Fig. \ref{Fig2} shows instead the results under mismatched conditions, now corresponding to a cosine squared $\cos^2 \theta = 0.46$: the KNN detector tends to behave as the AMF and, hence, it is more robust than Kelly's detector, without losing the CFAR property. This is a remarkable result, since classical CFAR  detectors usually trade-off robustness for some $P_d$ loss under matched conditions.

\begin{figure}
\centering
\includegraphics[width=0.7\textwidth]{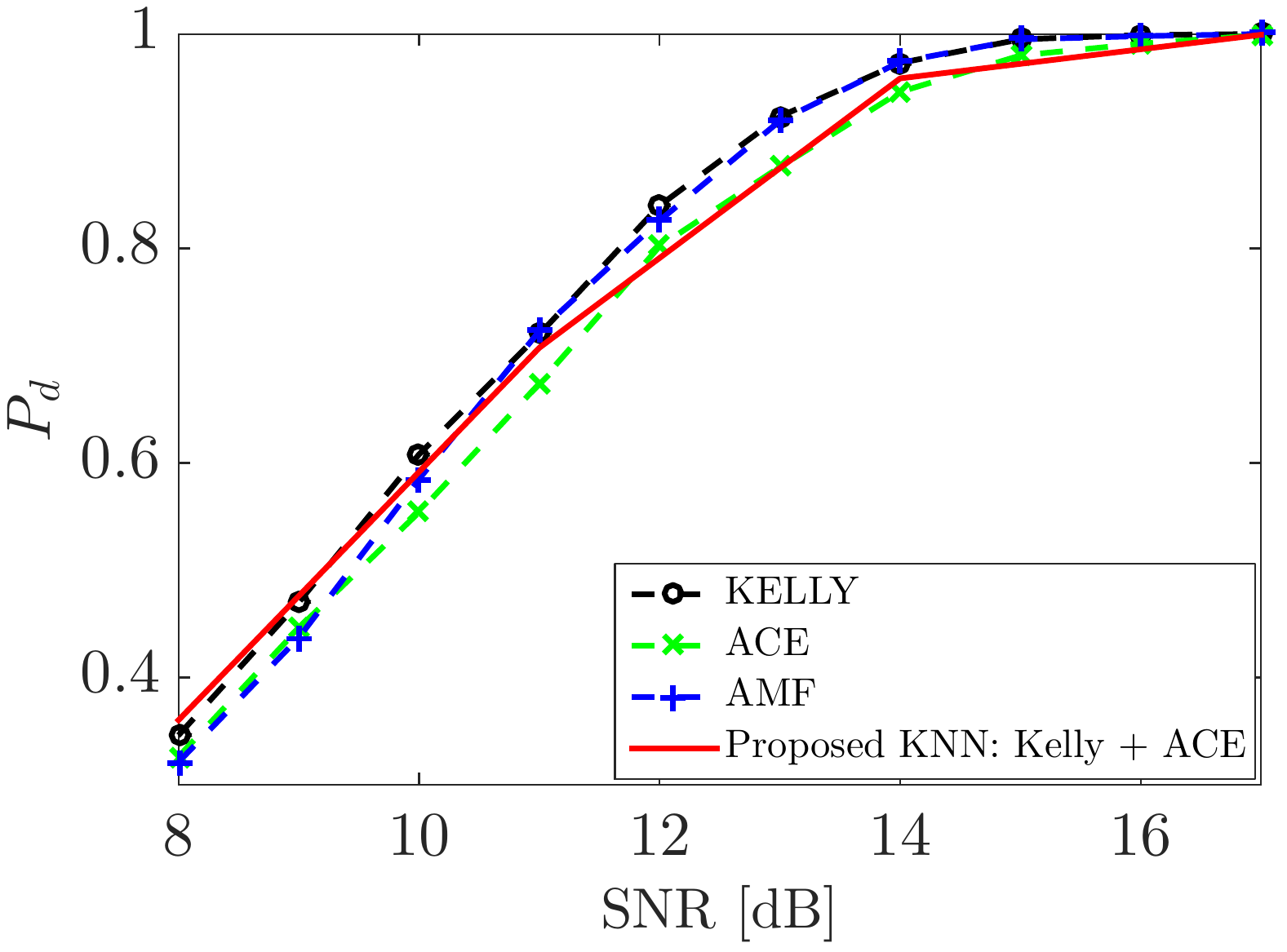} 
\caption{$P_d$ vs SNR under matched conditions
for a KNN fed by eq. (\ref{eq:CFARfeat1}).}
\label{Fig5}
\end{figure}

\begin{figure}
\centering
\includegraphics[width=0.7\textwidth]{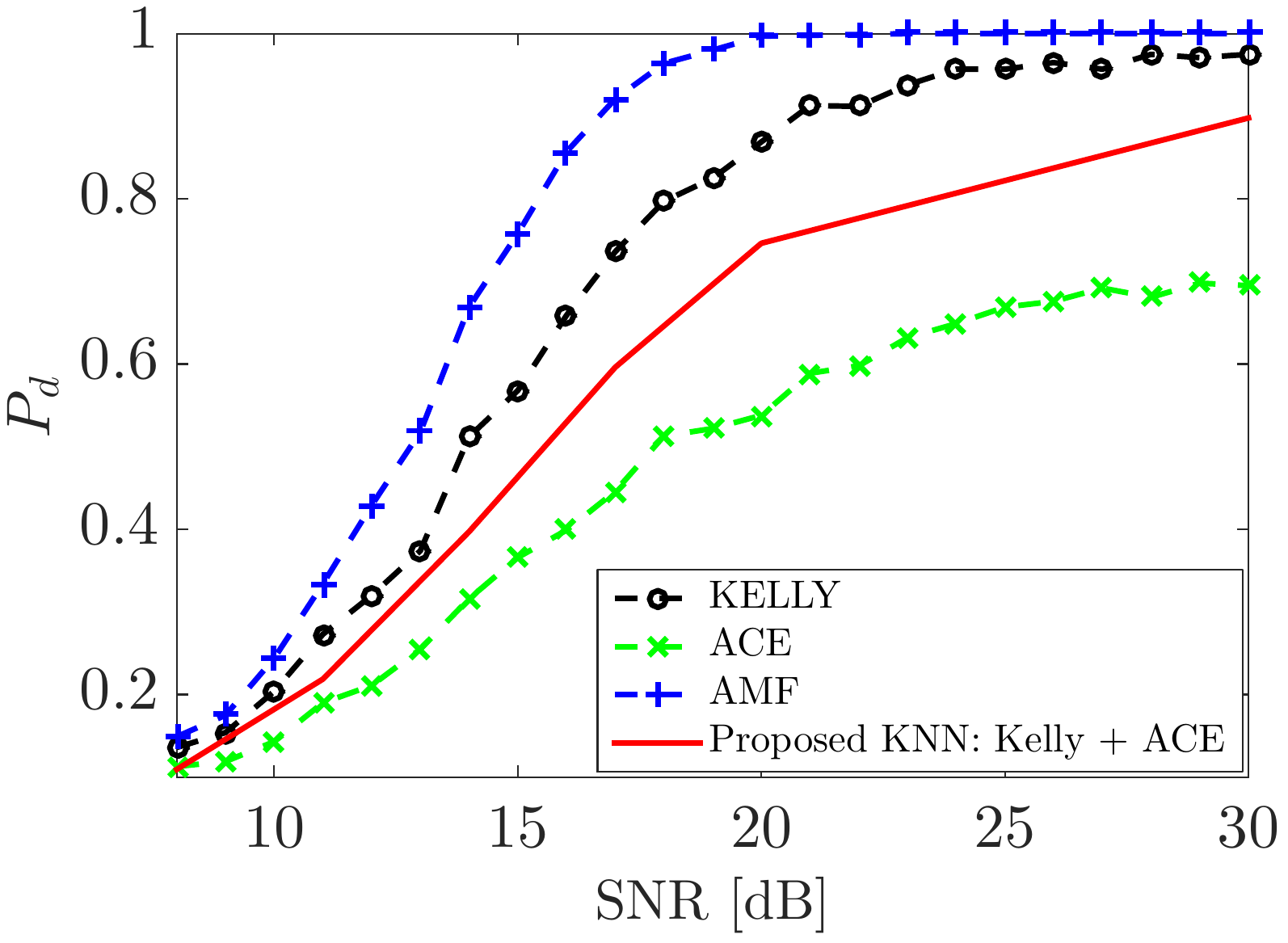} 
\caption{$P_d$ vs SNR under mismatched conditions with $\cos^2 \theta = 0.46$, for a KNN fed by eq. (\ref{eq:CFARfeat1}).}
\label{Fig6}
\end{figure}

As a further example, we illustrate in the following the performance of a detector whose feature vector is composed by the Kelly's detector and ACE, namely \cite{BOR-MorganClaypool}
\be
\bx = \left[d_1\tilde{t} \ \ d_2 \frac{\tilde{t}}{1-\beta} \right]^T. \label{eq:CFARfeat1}
\ee
In this case, the aim is to obtain a detector with intermediate performance between such two receivers, and again without loss of generality, we consider $d_1=1$ while we set $d_2=0.8$; the remaining parameters are as before. Fig. \ref{Fig5} shows that the KNN-based detector has indeed intermediate performance between Kelly's detector and ACE, being almost coincident with the former until an SNR of about 11 dB, while for higher SNR it experiences the same $P_d$ loss of the latter. Also the performance under mismatched conditions are intermediate between the two, as shown in Fig. \ref{Fig6}, thus yielding as a whole a novel detector with  a behavior different from the classical receivers. This is more  generally the advantage of the proposed approach, in which further compositions of the feature vector and/or tuning of the parameters can be investigated to design detectors with innovative characteristics.

\section{Conclusion}\label{sec:conclusion}
A novel  approach to the design of radar detectors, based on the KNN algorithm, has been investigated. $P_{fa}$ and $P_d$ have been analytically characterized in closed-form  for the general case.
This result has been then used to prove the CFAR property and to identify the relevant performance parameters of a KNN-based detector having as feature vector an arbitrary number of well-known radar detection statistics.
Some design examples have been reported: for instance, it is possible to obtain a detector that is as powerful as Kelly's detector (under matched conditions) while being more robust (under mismatched conditions) by considering only a two-dimensional feature vector; also, a  receiver with intermediate selectivity can be obtained by changing one of such features. 
A further design example has been reported to illustrate the effectiveness of the idea: in particular, using raw data as feature vector, it is possible to obtain a receiver that is more powerful than Kelly's detector, though it does not possess the CFAR property. 
As a whole, such results reveal that novel types of radar detectors can be designed through the proposed approach, and several possibilities remain to be investigated, potentially delivering solutions that cannot be obtained through conventional design tools for radar detection.
 
\section*{Appendix A}
\renewcommand{\theequation}{A.\arabic{equation}}

\section*{Proof of proposition 2}

In proposition 1,
$P(\overline{\ell} > T)$ and, hence, $P_{fa}$
if $\bx$ 
comes from the $H_0$ hypothesis,
is expressed in terms of $p_0 \left( \bx, \bm{x}^0_{1:k-M} \right)$ and $p_1 \left( \bx, \bm{x}^1_{1:N_T-M} \right)$.

We will specialize the formulas for $p_0$ and $p_1$
by considering
a feature vector of the form \eqref{eq::vect_struct}.
To this end, it will be immediately apparent that the involved RVs
are the independent random quantities
$(\tilde{t}$, $\beta)$, $(\tilde{t}^0,\beta^0)$, $(\tilde{t}^1,\beta^1)$, $(\tilde{t}^0_i,\beta^0_i)$s, and $(\tilde{t}^1_i,\beta^1_i)$s.
To this end, we also notice that \cite{BOR-MorganClaypool}
\begin{itemize}
\item
$\tilde{t}$ (under $H_0$), $\tilde{t}^0$ and $\tilde{t}_{ i}^0$s
are RVs ruled by the complex central $F$-distribution with 1 and $K-N+1$ complex degrees of freedom and are independent of $\beta$, $\beta^0$ and $\beta_{i}^0$s, respectively;
we write, $\tilde{t}, \tilde{t}^0, \tilde{t}_{i}^0 \sim  \cC\cF_{1,K-N+1}$;
\item
$\beta$, $\beta^0$ and $\beta_{i}^0$s are complex central beta RVs with  $K-N+2$ and $N-1$ complex degrees of freedom, in symbols $\beta, \beta^0, \beta_{i}^0 \sim\cC\beta_{K-N+2,N-1}$;
\item
$\tilde{t}_{ i}^1$ given  $\beta_{i}^1$ and $\tilde{t}^1$ given $\beta^1$ are ruled by a complex non-central
$F$-distribution with 1 and $K-N+1$ complex degrees of freedom
with non-centrality parameter $\delta$, i.e.,
$\tilde{t}_{i}^1,\tilde{t}^1 \sim \cC\cF_{1,K-N+1}(\delta)$ with $\delta^2=\mbox{SNR} \beta_{i}^1$ or $\delta^2=\mbox{SNR} \beta^1$, respectively;
\item
$\beta_{i}^1$ and $\beta^1$ are complex central beta RVs
with $K-N+2$ and $N-1$ complex degrees of freedom, i.e., $\beta_{i}^1, \beta^1 \sim \cC\beta_{K-N+2,N-1}$.
\end{itemize}


As a matter of fact, we have that
\begin{align*}
& p_0 \left( \bx, \bm{x}^0_{1:k-M} \right) 
 =
P \left( \left\{\|\bx^0-\bx\| \geq \max_{r \in \{1, \ldots, k-M \}} \|\bx_r^0-\bx\| \right\} \Big| \bx, \bm{x}^0_{1:k-M} \right) 
\\ &=
P \! \left( \! \left\{ \!  \sqrt{\sum_{j=1}^m d^2_j \left(\tilde{t}^0b^0[j] - \tilde{t}b[j]  \right)^2}
\geq \!\!\!\!
\max_{r \in \{1, \ldots, k-M \}} \! \|\bx_r^0-\bx\| \! \right\} \!\! \Big| \bx, \bm{x}^0_{1:k-M} \right)
\\ &=
P \left(  \left\{ \sum_{j=1}^m d^2_j \left(\tilde{t}^0b^0[j] - \tilde{t}b[j]  \right)^2 
\geq
c^2 \left( \bx, \bm{x}^0_{1:k-M}\right) \right\} \Big| \bx, \bm{x}^0_{1:k-M} \right).
\end{align*}
Moreover, the inequality
$$
\sum_{j=1}^m d^2_j \left(\tilde{t}^0b^0[j] - \tilde{t}b[j]  \right)^2
\geq
c^2 \left( \bx, \bm{x}^0_{1:k-M}\right)
$$
can be re-written as
$$
\gamma^0_1 (\tilde{t}^0)^2 - 2\tilde{t}\gamma^0_2\tilde{t}^0 + \gamma^0_3\tilde{t}^2 - c^2\left( \bx, \bm{x}^0_{1:k-M}\right) \geq 0
$$
where $\gamma^0_1$, $\gamma^0_2$, and $\gamma^0_3$ are proper coefficients that incorporate all the elements resulting from the square and product operations. Given $\beta^0$, the above inequality can be solved for
$\tilde{t}^0$ as
$$
\tilde{t}^0
\in
\left\{
\begin{array}{lll}
(-\infty,r_1) \cup  (r_2,+\infty), & \mbox{if} & \beta^0: \Delta_0 >0 \\
\R, & \mbox{if} & \beta^0: \Delta_0 <0
\end{array}
\right.
$$
with $\Delta_0$ the discriminant of the corresponding quadratic equation and $r_1,r_2$, $r_1<r_2$ its real and distinct roots (assuming $\Delta_0 >0$). Obviously, $r_1,r_2$ are continuous functions of $\bx, \bm{x}^0_{1:k-M}$, and 
$\beta^0$.
Notice also that the equation always admits a root whose real part is positive; it follows that
we have to consider the cases
$r_1<0 < r_2$ and
$0< r_1< r_2$. Thus, letting $I_{n_0}\left( \bx, \bm{x}^0_{1:k-M}\right)=\{ \beta^0 \in (0,1): \Delta_0 >0, r_1<0 \}$,
$I_{p_0}\left( \bx, \bm{x}^0_{1:k-M} \right)=\{ \beta^0 \in (0,1): \Delta_0 >0, r_1>0 \}$, and exploiting the fact that $\tilde{t}^0$ is a nonnegative random variable,
we have that
\begin{align*}
 p_0 \left( \bx, \bm{x}^0_{1:k-M}\right) \!&=\!
\int_{I_{n_0}} \!\!\!\!
P \left(  \tilde{t}^0 >r_2   | \bx, \bm{x}^0_{1:k-M}, \beta^0 \right)
f_{\beta^0}(\beta^0) d \beta^0
\\ &+
\int_{I_{p_0}} \!\!\!\!
P \left(  \tilde{t}^0 \in (0,r_1)\cup (r_2,+\infty)   | \bx, \bm{x}^0_{1:k-M}, \beta^0 \right)
f_{\beta^0}(\beta^0) d \beta^0
\\ &+
\int_{(0,1)\setminus \{I_{n_0} \cup I_{p_0}\}} \!\!\!\!
f_{\beta^0}(\beta^0) d \beta^0
\\ &=
\int_{I_{n_0}} 
\left( 1-F_{\tilde{t}^0}(r_2) \right)
f_{\beta^0}(\beta^0) d \beta^0
\\ &+
\int_{I_{p_0}}
\left(
F_{\tilde{t}^0}(r_1) +1 - F_{\tilde{t}^0}(r_2)
\right)
f_{\beta^0}(\beta^0) d \beta^0
\\ &+
\int_{(0,1)\setminus \{I_{n_0} \cup I_{p_0}\}}
f_{\beta^0}(\beta^0) d \beta^0
\end{align*}
where 
$F_{\tilde{t}^0}$ is the CDF of $\tilde{t}^0$
while $f_{\beta^0}$ is the PDF of $\beta^0$.

Similarly, we have that
\begin{align*}
 & p_1 \left( \bx, \bm{x}^1_{1:N_T-M} \right) 
 =
P \left( \left\{ \|\bx^1 - \bx\| \leq \min_{r \in \{1, \ldots, N_T-M \}} \|\bx_r^1,\bx\| \right\} | \bx, \bm{x}^1_{1:N_T-M} \right) 
\\ &=
P \left( \left\{ \sqrt{\sum_{j=1}^m d^2_j \left(\tilde{t}^1b^1[j] - \tilde{t}b[j] \right)^2}
\leq
a^2 \left( \bx, \bm{x}^1_{1:N_T-M}\right) \right\}| \bx, \bm{x}^1_{1:N_T-M} \right) 
\\ &=
P \left( \left\{ \sum_{j=1}^m d^2_j \left(\tilde{t}^1b^1[j] - \tilde{t}b[j] \right)^2
\leq
a^2 \left( \bx, \bm{x}^1_{1:N_T-M}\right) \right\} | \bx, \bm{x}^1_{1:N_T-M} \right).
\end{align*}
Again, the inequality
$$
\sum_{j=1}^m d^2_j \left(\tilde{t}^1b^1[j] - \tilde{t}b[j] \right)^2
\leq
a^2 \left( \bx, \bm{x}^1_{1:N_T-M}\right) 
$$
can be re-written
as
$$
\gamma^1_1
\left( \tilde{t}^1 \right)^2
-2 \tilde{t} \gamma^1_2 \tilde{t}^1
+ \gamma^1_3 - a^2 \left( \bx, \bm{x}^1_{1:N_T-M}\right) \leq 0,
$$
with $\gamma^1_1$, $\gamma^1_2$ and $\gamma^1_3$  proper coefficients. Given $\beta^1$,
the above inequality can be solved for $\tilde{t}^1$ as
$$
\tilde{t}^1
\in
\left\{
\begin{array}{lll}
(r_3, r_4), & \mbox{if} & \beta^1: \Delta_1 >0 \\
\emptyset, & \mbox{if} & \beta^1: \Delta_1 <0
\end{array}
\right.
$$
with $\Delta_1$ the discriminant of the corresponding quadratic equation and $r_3,r_4$, $r_3<r_4$ its real and distinct roots (assuming $\Delta_1 >0$). $r_3,r_4$ are continuous functions of $\bx, \bm{x}^1_{1:N_T-M}, \beta^1$.
Again, the equation always admits a root whose real part is positive; therefore, we should consider the cases
$r_3<0 < r_4$ and
$0< r_3< r_4$. Thus, letting $I_{n_1}\left( \bx, \bm{x}^1_{1:N_T-M} \right)=\{ \beta^1 \in (0,1): \Delta_1 >0, r_3<0 \}$,
$I_{p_1}\left( \bx, \bm{x}^1_{1:N_T-M} \right)=\{ \beta^1 \in (0,1): \Delta_1 >0, r_3>0 \}$, and exploiting the fact that $\tilde{t}^1$ is a nonnegative random variable,
it follows that
\begin{align*}
p_1 \left( \bx, \bm{x}^1_{1:N_T-M} \right) &=
\int_{I_{n_1}}
P \left( \tilde{t}^1 \in (0, r_4)   | \bx,\bm{x}^1_{1:N_T-M}, \beta^1 \right)
f_{\beta^1}(\beta^1) 
d \beta^1
\\ &+
\int_{I_{p_1}}
P \left(  \tilde{t}^1 \in (r_3, r_4)  | \bx, \bm{x}^1_{1:N_T-M}, \beta^1 \right)
f_{\beta^1}(\beta^1) 
d \beta^1
\\ &=
\int_{I_{n_1}}
F_{\tilde{t}^1|\beta^1}(r_4) 
f_{\beta^1}(\beta^1) 
d \beta^1
\\ &+
\int_{I_{p_1}}
\left( F_{\tilde{t}^1|\beta^1}(r_4) - F_{\tilde{t}^1|\beta^1}(r_3) \right)
f_{\beta^1}(\beta^1) 
d \beta^1
\end{align*}
where
$F_{\tilde{t}^1|\beta^1}$ is the CDF of $\tilde{t}^1$ given $\beta^1$,
and $f_{\beta^1}$ is the PDF of $\beta^1$. This concludes the proof of the expressions given in the statement.

Following the lead of previous reasoning, it is also possible to determine the parameters $P_d$ depends on, also under mismatched conditions (i.e., a steering vector not aligned with the nominal one).
In fact, 
to compute $P_d$ we
suppose that $\bx$ comes from a 
$\bz$
containing signal plus noise and,
in particular
$$
\bz = \alpha \bp + \bn
$$
with $\bp$ not necessarily aligned with $\bv$.
It follows that
\begin{eqnarray*}
P_d &=&
1-
\left(
\begin{array}{c}
N_T \\ k-M
\end{array}
\right)
\left(
\begin{array}{c}
N_T \\ N_T-M
\end{array}
\right)
\\ &\times&
E_{ \by }\left[ I_{\cal Y}(\by)
\left( p_0 \left( \bx, \bx_{1:k-M}^0 \right) \right)^{N_T-k+M}
\left( p_1 \left( \bx, \bx_{1:N_T-M}^1 \right) \right)^{M}
\right]
\end{eqnarray*}
where
again
the $(\tilde{t}$, $\beta)$, $(\tilde{t}^0,\beta^0)$, $(\tilde{t}^1,\beta^1)$, $(\tilde{t}^0_i,\beta^0_i)$s, and $(\tilde{t}^1_i,\beta^1_i)$s are independent random quantities. However, this time
\cite{BOR-MorganClaypool}
\begin{itemize}
\item
$\tilde{t}^0, \tilde{t}_{ i}^0 \sim  \cC\cF_{1,K-N+1}$ and
$\beta^0, \beta_{i}^0 \sim\cC\beta_{K-N+2,N-1}$;
\item
given $\beta^1$ or $\beta_{i}^1$, $\tilde{t}^1, \tilde{t}_i^1 \sim \cC\cF_{1,K-N+1}(\delta)$
with non-centrality parameter $\delta$, where $\delta^2=\mbox{SNR} \beta^1$ or $\delta^2=\mbox{SNR} \beta_i^1$, respectively, while
$\beta^1, \beta_{i}^1 \sim \cC\beta_{K-N+2,N-1}$;
\item
given $\beta$, $\tilde{t} \sim \cC\cF_{1,K-N+1}(\delta)$ 
where $\delta$ is the non-centrality parameter, with $\delta^2=\mbox{SNR}_p \beta \cos^2 \theta$ and 
$
\mbox{SNR}_p$, $
\cos^2 \theta$ defined in the  statement;
\item
$\beta$ is a complex non-central beta RV with $K-N+2$ and $N-1$ complex degrees of freedom
with non-centrality parameter $\delta$, $\delta^2=\mbox{SNR}_p \sin^2\theta$,
in symbols $\beta \sim \cC\beta_{K-N+2,N-1}(\delta)$.
\end{itemize}


\section*{References}


\begin{thebibliography}{}

\bibitem{Kelly}
E. J. Kelly, ``An Adaptive Detection Algorithm,'' {\em IEEE Trans. Aerosp. and Electron. Syst.}, Vol. 22, No.~2, Mar.~1986.


\bibitem{Kelly89}
E. J. Kelly, ``Performance of an Adaptive
Detection Algorithm; Rejection
of Unwanted Signals,''
{\em IEEE Trans. Aerosp. and Electron. Syst.},
Vol. 25, No.~2, Mar.~1989.


\bibitem{Kelly-Nitzberg}
F.~C.~Robey, D.~L.~Fuhrman, E.~J.~Kelly, R.~Nitzberg,
``A CFAR Adaptive Matched Filter Detector,''
{\em IEEE Trans. Aerosp. and Electron. Syst.},
Vol.~29, No.~1, Jan.~1992.

\bibitem{Asymptotically}
E. Conte, M. Lops, G. Ricci, ``Asymptotically Optimum Radar Detection
in Compound Gaussian Noise,'' {\em IEEE Trans. Aerosp. and Electron. Syst.},
Vol.~31, No. 2, April~1995.



\bibitem{ACE}
S. Kraut, L. L. Scharf, ``The CFAR adaptive subspace detector is
a scale-invariant GLRT,'' {\em IEEE Trans. Signal Process.}, Vol.~47, No. 9, Sept. 1999.


\bibitem{Pulsone-Rader}
N. B. Pulsone, C. M. Rader, ``Adaptive Beamformer Orthogonal
Rejection Test,'' {\em IEEE Trans. Signal Process.}, Vol.~49,
No.~3, Mar. 2001.

\bibitem{Fabrizio-Farina}
G. A. Fabrizio, A. Farina, M. D. Turley, ``Spatial Adaptive Subspace Detection in OTH Radar,''
{\em IEEE Trans. Aerosp. and Electron. Syst.},
Vol. 39, No. 4, Oct. 2003.

\bibitem{W-ABORT}
F.~Bandiera, O.~Besson, G.~Ricci, ``An ABORT-Like Detector
With Improved Mismatched Signals Rejection Capabilities,''
{\em IEEE Trans. Signal Process.}, Vol. 56, No. 1, Jan. 2008.

\bibitem{CR_SPL}
A. Coluccia, G. Ricci,
``A Tunable W-ABORT-like Detector with Improved Detection  vs Rejection Capabilities Trade-Off,''
 {\em IEEE Signal Process. Lett.}, Vol. 22, No. 6, Jun. 2015.
 
 \bibitem{Besson}
O. Besson, A. Coluccia, E. Chaumette, G. Ricci, F. Vincent,
``Generalized likelihood ratio test for detection of Gaussian rank-one signals 
in Gaussian noise with unknown statistics,''
{\em IEEE Trans. Signal Process.}, Vol. 65, No. 4, 15 Feb. 2017.

\bibitem{Ball}
J. E. Ball, ``Low signal-to-noise ratio radar target detection using Linear Support Vector Machines (L-SVM),'' {\em 2014 IEEE Radar Conference}, Cincinnati, OH, USA, 
19-23 May 2014.

\bibitem{LeiouWang}
L. Wang, D. Wang, C. Hao,
``Intelligent CFAR Detector Based on Support Vector Machine,''
{\em IEEE Access}, Vol.~5, 16 Nov. 2017.


\bibitem{SVM_JSAC}
K.M. Thilina, K.W. Choi, N. Saquib, E. Hossain, ``Machine learning techniques for cooperative spectrum sensing in cognitive radio networks,'' {\em IEEE J.  Sel. Areas in Commun.}, Vol. 31, No. 11, Nov 2013.

\bibitem{SVM_TSP}
Y. Xu, P. Cheng, Z. Chen, Y. Li, B. Vucetic, ``Mobile col- laborative spectrum sensing for heterogeneous networks: A Bayesian machine learning approach,'' 
{\em IEEE Trans. Signal Process.}, Vol. 66, No. 21, Nov 2018.

\bibitem{EUSIPCO2019}
A. Coluccia, A. Fascista, G. Ricci, ``Spectrum sensing by higher-order SVM-based detection'', {\em 27th European Signal Processing Conference (EUSIPCO 2019)}, A Coru\~na, Spain, 2-6 Sep. 2019.

\bibitem{CR_Boston2019}
A. Coluccia, G. Ricci,
``Radar detection in K-distributed clutter plus thermal noise based on KNN methods,''
{\em IEEE Radar Conference}, Boston, MA, USA, 22-26 Apr. 2019.

\bibitem{detection_radarcon19}
D. Brodeski, I. Bilik, R. Giryes, ``Deep radar detectors,''
{\em IEEE Radar Conference}, Boston, MA, USA, 22-26 Apr. 2019.

\bibitem{classification_radarcon19}
K. Patel, K. Rambach, T. Visentin, D. Rusev, M. Pfeiffer, B. Yang,
``Deep learning-based object classification on automotive radar spectra,''
{\em IEEE Radar Conference}, Boston, MA, USA, 22-26 Apr. 2019.

\bibitem{waveform_design_radarcon19}
J. M. Kurdzo, J. Y. N. Cho, B. L. Cheong, R. D. Palmer,
``A neural network approach for waveform generation and selection with multi-mission radar,''
{\em IEEE Radar Conference}, Boston, MA, USA, 22-26 Apr. 2019.

 \bibitem{BOR-MorganClaypool}
F. Bandiera, D. Orlando, G. Ricci,
``Advanced Radar Detection Schemes Under Mismatched Signal Models,''
{\em Synthesis Lectures on Signal Processing No. 8, Morgan \& Claypool Publishers},
2009.

\end{thebibliography}
\end{document}